\documentclass[envcountsame]{llncs}
\usepackage{amsmath,amssymb}
\usepackage[utf8]{inputenc}
\usepackage{todonotes}

\usepackage{thmtools}
\usepackage{chngcntr}
\usepackage{shuffle}

\def\N{\mathbb{N}}
\def\R{\mathbb{R}}

\def\ta{\mathtt{a}}
\def\tb{\mathtt{b}}
\def\tc{\mathtt{c}}

\def\tn{\mathtt{n}}
\def\tx{\mathtt{x}}
\def\ty{\mathtt{y}}
\def\nth#1{#1$^{\text{th}}$}

\title{Reconstructing Words from Right-Bounded-Block Words}
\author{Pamela Fleischmann\inst{1} \and Marie Lejeune\inst{2}\thanks{Supported by a FNRS fellowship.} \and Florin Manea\inst{3}\thanks{Supported by the DFG grant MA 5725/2-1.} \and Dirk 
Nowotka\inst{1} \and Michel Rigo\inst{2}}
\authorrunning{P. Fleischmann \and M. Lejeune \and F. Manea \and D. Nowotka \and M. Rigo}

\institute{Kiel University, 
Germany \email{\{fpa,dn\}@informatik.uni-kiel.de} \and 
University of Liège, Belgium \email{\{m.lejeune,m.rigo\}@uliege.be} \and University of Göttingen, Germany \email{florin.manea@informatik.uni-goettingen.de}}

\begin{document}

\maketitle
\begin{abstract}
  A reconstruction problem of words from scattered factors asks for the minimal information, like multisets of scattered factors of a given length or the number of occurrences of scattered factors from a given set, necessary to uniquely determine a word. We show that a word $w\in\{\ta,\tb\}^*$ can be reconstructed from the number of occurrences of at most $\min(|w|_\ta,|w|_\tb)+1$ scattered factors of the form $\ta^i \tb$, where $|w|_{\ta}$ is the number of occurrences of the letter $\ta$ in $w$. Moreover, we generalize the result to alphabets of the form $\{1, \ldots, q\}$ by showing that at most $\sum_{i=1}^{q-1} |w|_i \, (q-i+1)$ scattered factors suffices to reconstruct $w$. Both results improve on the upper bounds known so far. Complexity time bounds on reconstruction algorithms are also considered here.

\end{abstract}

\section{Introduction}
The general scheme for a so-called {\em reconstruction problem} is the following one: 
given a sufficient amount of information about substructures of a hidden discrete structure, 
can one uniquely determine this structure? In particular, what are the fragments about the structure needed to recover it all. For instance, a square matrix of size at least~$5$ can be reconstructed from its principal minors given in any order \cite{Manvel}.

In graph theory, given some subgraphs of a graph (these subgraphs may share some common vertices and edges), can one uniquely rebuild the original graph? 
Given a finite undirected graph $G=(V,E)$ with $n$ vertices, consider the multiset made of the $n$ induced subgraphs of $G$ obtained by deleting exactly one vertex from $G$. In particular, one knows how many isomorphic subgraphs of a given class appear. Two graphs leading to the same multiset (generally called a {\em deck}) are said to be {\em hypomorphic}. A conjecture due to Kelly and Ulam states that two hypomorphic graphs with at least three vertices are isomorphic \cite{Kelly,Ulam}. A similar conjecture in terms of edge-deleted subgraphs has been proposed by Harary \cite{Har}. These conjectures are known to hold true for several families of graphs.

A {\em finite word}, i.e., a finite sequence of letters of some given alphabet, can be seen as an edge- or vertex-labeled linear tree. So variants of the graph reconstruction problem can be considered and are of independent interest.  
Participants of the Oberwolfach meeting on Combinatorics on Words in 2010 \cite{Oberwolfach} gave a list 
 of 18 important open problems in the field. 
Amongst them, the twelfth problem is stated as {\em reconstruction from subwords of given length}. 
In the following statement and all along the paper, a {\em subword} of a word is understood as a subsequence of not necessarily contiguous letters from this word, i.e., subwords can be obtained by deleting letters from the given word. To highlight this latter property, they are often called {\em scattered subwords} or {\em scattered factors}, which is the notion we are going to use.

\begin{definition}
Let $k,n$ be natural numbers. Words of length $n$ over a given alphabet are said to be {\em $k$-reconstructible} whenever the multiset of scattered factors of length~$k$ (or {\em $k$-deck}) uniquely determines any word of length~$n$.
\end{definition} 

Notice that the definition requires multisets to store the information how often a scattered factor occurs in the words. For instance, the scattered factor $\tb\ta$ occurs three times in $\tb\ta\tb\ta$ which provides more information for the reconstruction than the mere fact that $\tb\ta$ is a scattered factor.

The challenge is to determine the function $f(n)=k$ where $k$ is the least integer for which words of length~$n$ are $k$-reconstructible.  
This problem has been studied by several authors and one of the first trace goes back to 1973 \cite{Kala}. 
Results in that direction have been obtained by M.-P.~Sch\"utzenberger (with the so-called {\em Sch\"utzenberger's Guessing game}) and L.~Simon \cite{Simon}.
They show that words of length~$n$ sharing the same multiset of scattered factors of length up to $\lfloor n/2 \rfloor +1$ are the same. Consequently, words of length~$n$ are $(\lfloor n/2 \rfloor +1)$-reconstructible. In \cite{KrRo}, this upper bound has been improved: Krasikov and Roditty have shown that words of length~$n$ are $k$-reconstructible for $k\ge\lfloor 16\sqrt{n} /7\rfloor+5$. 
On the other hand Dudik and Schulmann \cite{Dudik} provide a lower bound: if words of length~$n$ are $k$-reconstructible, then $k\ge 3^{(\sqrt{2/3}-o(1))\log_3^{1/2}n}$. Bounds were also considered in \cite{Stock}. Algorithmic complexity of the reconstruction problem is discussed, for instance, in \cite{DErdos}. Note that the different types of reconstruction problems have application in philogenetic networks, see, e.g., \cite{Philo}, or in the context of molecular genetics \cite{Molec} and coding theory \cite{coding}.

Another motivation, close to combinatorics on words, stems from the study of $k$-binomial equivalence of finite words and $k$-binomial complexity of infinite words (see \cite{RiSa15} for more details). 
Given two words of the same length, they are $k$-binomially equivalent if they have the same multiset of scattered factors of length $k$, also known as {\em $k$-spectrum} (\cite{BerKar03}, \cite{conf/dlt/Manuch99}, \cite{RozSal97}). Given two words $x$ and $y$ of the same length, one can address the following problem: decide whether or not $x$ and $y$ are $k$-binomially equivalent? A polynomial time decision algorithm based on automata and a probabilistic algorithm have been addressed in \cite{decision}. A variation of our work would be to find, given $k$ and $n$, a minimal set of scattered factors for which the knowledge of the number of occurrences in $x$ and $y$ permits to decide $k$-binomial equivalence. 

Over an alphabet of size $q$, there are $q^k$ pairwise distinct length-$k$ factors. 
If we relax the requirement of only considering scattered factors of the same length, 
another interesting question is to look for a minimal (in terms of cardinality) multiset of scattered factors to reconstruct entirely a word. 
Let the {\em binomial coefficient} $\binom{u}{x}$ be the number of occurrences of $x$ as a scattered factor of $u$. The general problem addressed in this paper is therefore the following one.

\begin{problem}\label{reconstproblem}
Let $\Sigma$ be a given alphabet and $n$ a natural number. We want to reconstruct a hidden word $w \in \Sigma^n$. To that aim, we are allowed to pick a word $u_i$ and ask questions of the type ``What is the value of $\binom{w}{u_i}$?". Based on the answers to questions related to $\binom{w}{u_1}, \ldots, \binom{w}{u_i}$, we can decide which will be the next question (i.e. decide which word will be $u_{i+1}$). 
We want to have the shortest sequence $(u_1, \ldots, u_k)$ uniquely determining $w$ by knowing the values of $\binom{w}{u_1}, \ldots, \binom{w}{u_k}$. 
\end{problem}
We naturally look for a value of $k$ less than the upper bound for $k$-reconstructibility. 

In this paper, we firstly recall the use of Lyndon words in the context of reconstructibility. A word $w$ over a totally ordered alphabet is called {\em Lyndon word} if it is the lexicographically smallest amongst all its rotations, i.e., $w=xy$ is smaller than $yx$ for all non trivial factorisations $w=xy$. Every binomial coefficient $\binom{w}{x}$ for arbitrary words $w$ and $x$ over the same alphabet can be deduced from the values of the coefficients $\binom{w}{u}$ for Lyndon words $u$ that are lexicographically less than or equal to $x$. This result is presented in Section~\ref{prel} along with the basic definitions. We consider an  alphabet  equipped with a total order on the letters. Words of the form $\ta^n \tb$ with letters $\ta<\tb$ and a natural number $n$ are a special form of Lyndon words, the so-called {\em right-bounded-block} words.

 We consider the reconstruction problem from the information given by the occurrences of right-bounded-block words as scattered factors of a word of length~$n$. In Section~\ref{binary} we show how to reconstruct a word uniquely from $m+1$ binomial coefficients of right-bounded-block words where $m$ is the minimum number of occurrences of $\ta$ and $\tb$ in the word. We also prove that this is less than the upper bound given in \cite{KrRo}. In Section~\ref{general} we reduce the problem for arbitrary finite alphabets $\{1,\dots,q\}$ to the binary case. Here we show that at most $\sum_{i=1}^{q-1} |w|_i \, (q-i+1) \le q|w|$ binomial coefficients suffice to uniquely reconstruct $w$ with $|w|_i$ being the number of occurrences of letter $i$ in $w$. Again, we compare this bound to the best known one for the classical reconstruction problem (from words of a given length). In the last section of the paper we also propose several results of algorithmic nature regarding the efficient reconstruction of words from given scattered factors.

\section{Preliminaries}\label{prel}
Let $\N$ be the set of natural numbers, $\N_0 = \N\cup\{0\}$, and let 
$\N_{\geq k}$ be the set of all natural numbers greater than or equal to $k$. 
Let $[n]$ denote the set $\{1,\ldots, n\}$ and $[n]_0 = [n]\cup\{0\}$
for an $n\in\N$.

An {\em alphabet} $\Sigma=\{\ta,\tb,\tc,\dots\}$ is a finite set of letters and 
a {\em word} is a finite sequence of letters. 
We let $\Sigma^*$ denote the set of all finite words over $\Sigma$. The {\em empty 
word} is denoted by $\varepsilon$ and 
$\Sigma^+$ is the free semigroup $\Sigma^*\backslash\{\varepsilon\}$.
The length of a word $w$ is denoted by $|w|$. 
Let 
$\Sigma^{\leq k}:=\{w\in\Sigma^{\ast}|\,|w|\leq k\}$ and $\Sigma^k$ be the 
set of all words of length exactly $k\in\N$. The number of occurrences of a 
letter $\ta\in\Sigma$ in a word $w\in\Sigma^{\ast}$ is denoted by $|w|_\ta$.
The \nth{$i$} letter of a word $w$ is given by $w[i]$ for $i\in[|w|]$.  The powers of 
$w\in\Sigma^{\ast}$ are 
defined recursively by $w^0=\varepsilon$, $w^n=ww^{n-1}$ for $n\in\N$.
A word $u\in\Sigma^{\ast}$ is a \emph{factor} of $w\in\Sigma^{\ast}$, if 
$w=xuy$ holds for some words $x,y\in\Sigma^{\ast}$. Moreover, $u$ is a 
\emph{prefix} of $w$ if $x=\varepsilon$ holds and a \emph{suffix} if 
$y=\varepsilon$ holds. The factor of $w$ from the \nth{$i$} to the \nth{$j$} 
letter will be denoted by $w[i..j]$ for $1\leq i\leq j\leq |w|$. Two words $u,v\in\Sigma^{\ast}$ are called {\em conjugates or rotations} of each other if there exist $x,y\in\Sigma^{\ast}$ with $u=xy$ and $v=yx$. Additional basic information about combinatorics on words can be found in \cite{lothaire}.

\begin{definition}
Let $<$ be a total ordering on $\Sigma$.
A word $w\in\Sigma^{\ast}$ is called {\em right-bounded-block word} if there exist 
$\tx,\ty\in\Sigma$ with $\tx<\ty$ and $\ell\in\N_0$ with $w=\tx^{\ell}\ty$.
\end{definition}

\begin{definition}
A word $u=\mathtt{a_1} \cdots \mathtt{a_n}\in\Sigma^n$, for $n\in\N$, is a {\em scattered factor} 
of 
a word 
$w\in\Sigma^+$ if there exist $v_0,\ldots,v_n\in\Sigma^{\ast}$ with 
$w=v_0\mathtt{a_1}v_1 \cdots 
v_{n-1}\mathtt{a_n}v_n$.
For words $w,u\in\Sigma^{\ast}$, define $\binom{w}{u}$ as the number of 
occurrences of $u$ as a scattered factor of $w$.
\end{definition}

\begin{remark}
Notice that $|w|_{\tx}=\binom{w}{\tx}$ for all $\tx\in\Sigma$.
\end{remark}

The following definition addresses Problem~\ref{reconstproblem}. 

\begin{definition}
A word $w\in\Sigma^n$ is called {\em uniquely reconstructible/determined} by 
the set $S \subset \Sigma^{\ast}$ if for all words 
$v\in\Sigma^n\backslash\{w\}$ there exists a word $u\in S$ with 
$\binom{w}{u}\neq\binom{v}{u}$.
\end{definition}

Consider $S=\{\ta\tb,\tb\ta\}$. Then $w=\ta\tb\tb\ta$ is not uniquely reconstructible by $S$ since $\left[\binom{w}{\ta\tb},\binom{w}{\tb\ta}\right]=[2,2]$ is also the $2$-vector of binomial coefficients of $\tb\ta\ta\tb$. On the other hand $S=\{\ta,\ta\tb,\ta\tb^2\}$ reconstructs $w$ uniquely. The following remark gives immediate results for binary alphabets.

\begin{remark}
Let $\Sigma=\{\ta,\tb\}$ and $w\in\Sigma^{n}$. If $|w|_\ta\in\{0,n\}$ then $w$ contains either only $\tb$ or $\ta$ and by the 
given length $n$ of $w$, $w$ is uniquely determined by $S=\{\ta\}$.
This fact is in particular an equivalence: $w\in\Sigma^{n}$ can be uniquely 
determined by $\{\ta\}$ iff $|w|_\ta\in\{0,n\}$. If $|w|_\ta\in\{1,n-1\}$, $w$ is not uniquely determined by $\{\ta\}$ as 
witnessed by $\ta\tb$ and $\tb\ta$ for $n=2$. It is immediately clear that the 
additional information $\binom{w}{\ta\tb}$ leads to unique determinism of $w$.
\end{remark}

Lyndon words play an important role regarding the reconstruction problem. As shown in \cite{Reut} only scattered factors which are Lyndon words are necessary to determine a word uniquely, i.e., $S$ can always be assumed to be a set of Lyndon words.

\begin{definition}
Let $<$ be a total ordering on $\Sigma$. A word $w\in\Sigma^{\ast}$  is a {\em Lyndon word} iff for all $u,v \in \Sigma^{+}$ with $w = uv$, we have $w <_{lex} vu$ where $<_{lex}$ is the lexicographical ordering on words induced by $<$.
\end{definition}

\begin{proposition}[\cite{Reut}]
Let $w$ and $u$ be two words. The binomial coefficient $\binom{w}{u}$ can be computed using only binomial coefficients of the type $\binom{w}{v}$ where $v$ is a Lyndon word of length up to $|u|$ such that $v \leq_{lex} u$.
\end{proposition}

To obtain a formula to compute the binomial coefficient $\binom{w}{u}$ for $w,u\in\Sigma^{\ast}$ by binomial coefficients $\binom{w}{v_i}$ for Lyndon words $v_1,\dots,v_k$ with $v_i\in\Sigma^{\leq|u|}$, $i\in[k]$, and $k\in\N$ the definitions of shuffle and infiltration are necessary \cite{lothaire}.

\begin{definition}\label{shuffle}
Let $n_1,n_2\in\N$, $u_1 \in \Sigma^{n_1}$, and $u_2 \in \Sigma^{n_2}$. Set $n=n_1+n_2$. The {\em shuffle} of $u_1$ and $u_2$ is the polynomial $u_1 \shuffle u_2 = \sum_{I_1,I_2} w(I_1,I_2)$
where the sum has to be taken over all pairs $(I_1,I_2)$ of sets that are partitions of $[n]$ such that $|I_1| = n_1$ and $|I_2| = n_2$.
If $I_1 = \{ i_{1,1} < \ldots < i_{1,n_1}\}$ and $I_2 = \{ i_{2,1} < \ldots < i_{2,n_2}\}$, then the word $w(I_1,I_2)$ is defined such that $w[i_{1,1}] w[i_{1,2}] \cdots w[i_{1,n_1}] = u_1$ and $w[i_{2,1}] w[i_{2,2}] \cdots w[i_{2,n_2}] = u_2$ hold.
\end{definition}

The infiltration is a variant of the shuffle in which equal letters can be merged.

\begin{definition}\label{infiltration}
Let $n_1,n_2\in\N$, $u_1 \in \Sigma^{n_1}$, and $u_2 \in \Sigma^{n_2}$. Set $n=n_1+n_2$. The {\em infiltration} of $u_1$ and $u_2$ is the polynomial $u_1 \downarrow u_2 = \sum_{I_1,I_2} w(I_1,I_2)$,
where the sum has to be taken over all pairs $(I_1,I_2)$ of sets of cardinality $n_1$ and $n_2$ respectively, for which the union is equal to the set $[n']$ for some $n' \leq n$. 
Words $w(I_1,I_2)$ are defined as in the previous definition. Note that some $w(I_1,I_2)$ are not well defined if $i_{1,j} = i_{2,k}$ but $u_1[j] \neq u_2[k]$. In that case they do not appear in the previous sum.
\end{definition}

Considering for instance $u_1=\ta\tb\ta$ and $u_2=\ta\tb$ gives the polynomials
\begin{align*}
u_1 \shuffle u_2 & = 2\ta\tb\ta\tb\ta+4\ta\ta\tb\tb\ta+2\ta\ta\tb\ta\tb+2\ta\tb\ta\ta\tb,\\
u_1 \downarrow u_2 &= \ta\tb\ta \shuffle \ta\tb + \ta\tb\ta + 2 \ta\tb\tb\ta + 2 \ta\ta\tb\ta + 2 \ta\tb\ta\tb.
\end{align*}

Based on Definitions~\ref{shuffle} and \ref{infiltration}, we are able to give a formula to compute a binomial coefficient from the ones making use of Lyndon words. This formula is given implicitely in \cite[Theorem~6.4]{Reut}: Let $u\in\Sigma^{\ast}$ be a non-Lyndon word. By \cite[Corollary~6.2]{Reut} there exist non-empty words $x,y\in\Sigma^{\ast}$ and with $u = xy$ and such that every word appearing in the polynomial $x \shuffle y$ is lexicographically less than or equal to $u$.
Then, for all word $w \in \Sigma^*$, we have
$$
 \binom{w}{u} = \frac{1}{(x \shuffle y, u)} \left[ \binom{w}{x} \binom{w}{y} - \sum_{v \in \Sigma^{\ast}, v\neq u} (x \downarrow y, v) \binom{w}{v}\right],
$$
where $(P, v)$ is a notation giving the coefficient of the word $v$ in the polynomial $P$. One may apply recursively this formula until only Lyndon factors are considered.

\begin{example}
Considering $\Sigma=\{\ta,\tb\}$ the binomial coefficient $\binom{w}{\tb \ta}$ can be computed using the Lyndon words $\ta$, and $\tb$ by 
\begin{align*}
\binom{w}{\tb\ta} &= \frac{1}{(\tb \shuffle \ta, \tb\ta)} \left[ \binom{w}{\tb} \binom{w}{\ta} - (\tb \downarrow \ta, \ta\tb) \binom{w}{\ta\tb} \right] 
= \binom{w}{\tb} \binom{w}{\ta} - \binom{w}{\ta\tb}.
\end{align*}
Regarding word length three, the Lyndon words are $\ta\ta\tb$ and $\ta\tb\tb$. Let us give formulas to compute $\binom{w}{\ta\tb\ta}$, $\binom{w}{\tb\ta\ta}$, $\binom{w}{\tb\ta\tb}$ and $\binom{w}{\tb\tb\ta}$.
Having $x=\ta\tb$ and $y=\ta$, we obtain
$$
\binom{w}{\ta\tb\ta} = \binom{w}{\ta\tb} \left[ \binom{w}{\ta} - 1 \right] - 2 \binom{w}{\ta\ta\tb}.
$$
\smallskip
For $u = \tb\ta\ta$, we can either choose $x = \tb$ and $y = \ta\ta$ or $x = \tb\ta$ and $y = \ta$. In the first case, we get
$$
\binom{w}{\tb\ta\ta} = \binom{w}{\tb}\binom{w}{\ta\ta} - \binom{w}{\ta\tb\ta} - \binom{w}{\ta\ta\tb}
$$
and by reinjecting formulas for $\binom{w}{\ta\ta}$ and $\binom{w}{\ta\tb\ta}$, obtained recursively,
$$
\binom{w}{\tb\ta\ta} = \left[ \binom{w}{\ta} - 1 \right] \left[ \frac{1}{2} \binom{w}{\ta} \binom{w}{\tb} - \binom{w}{\ta\tb} \right] + \binom{w}{\ta\ta\tb}.
$$
\smallskip 
Finally, the last two formulas are quite similar to what we already had:
$$
\binom{w}{\tb\ta\tb} = \binom{w}{\ta\tb} \left[ \binom{w}{\tb} - 1 \right] - 2 \binom{w}{\ta\tb\tb}
$$
and
$$
\binom{w}{\tb\tb\ta} = \left[ \binom{w}{\tb} - 1 \right] \left[ \frac{1}{2} \binom{w}{\ta} \binom{w}{\tb} - \binom{w}{\ta\tb} \right] + \binom{w}{\ta\tb\tb}.
$$
\end{example}

\section{Reconstruction from Binary Right-Bounded-Block Words }\label{binary}
In this section we present a method to reconstruct a binary word uniquely from binomial coefficients of right-bounded-block words. Let $n\in\N$ be a natural number and $w\in\{\ta,\tb\}^n$ a word. Since the word length $n$ is assumed to be known, $|w|_\ta$ is known if $|w|_\tb$ is given and vice versa. Set for abbreviation $k_u=\binom{w}{u}$ for $u\in \Sigma^{\ast}$. Moreover we assume w.l.o.g. $k_\ta\leq k_\tb$ and that $k_\ta$ is known (otherwise substitute each $\ta$ by $\tb$ and each $\tb$ by $\ta$, apply the following reconstruction method and revert the substitution). This implies that $w$ is of the form
\begin{align}\label{basisword}
\tb^{s_1}\ta\tb^{s_2}\dots \tb^{s_{k_\ta}}\ta\tb^{s_{k_\ta+1}}
\end{align}
for $s_i\in\N_0$ and $i\in[|w|_{\ta}+1]$ with $\sum_{i\in[k_\ta +1]}s_i=n-k_\ta=k_\tb$ and thus we get
for $\ell\in[k_\ta]_0$ 
\begin{equation}\label{alb}
k_{\ta^{\ell}\tb}=\binom{w}{\ta^{\ell}\tb}=\sum_{i=\ell+1}^{k_\ta+1}\binom{i-1}{\ell}s_i.
\end{equation}

\begin{remark}\label{ci}
Notice that for fixed 
$\ell\in[k_\ta]_0$ and $c_i=\binom{i-1}{\ell}$ for 
$i\in[k_\ta+1]\backslash[\ell]$, we have $c_i<c_{i+1}$
and especially $c_{\ell+1}=1$ and $c_{\ell+2}=\ell+1$.
\end{remark}

Equation~(\ref{alb}) shows that reconstructing a word uniquely from binomial coefficients of 
right-bounded-block words equates to solve a system of Diophantine equations. 
The knowledge of $k_{\tb}, \ldots, k_{\ta^{\ell} \tb}$ provides $\ell + 1$ equations. If the equation of $k_{\ta^{\ell} \tb}$ has a unique solution for $\{ s_{\ell + 1}, \ldots, s_{k_{\ta} + 1} \}$ (in this case we say, by language abuse, that $k_{\ta^{\ell} \tb}$ is {\em unique}), then the system in row echelon form has a unique solution and thus the binary word is uniquely reconstructible.
Notice that $k_{\ta^{k_\ta}\tb}$ is always unique since 
$k_{\ta^{k_\ta}\tb}=s_{k_{\ta}+1}$. 

Consider $n=10$ and $k_\ta=4$. This leads to $w=\tb^{s_1}\ta \tb^{s_2}\ta \tb^{s_3}\ta \tb^{s_4}\ta\tb^{s_5}$ with $\sum_{i\in[5]}s_i=6$. Given $k_{\ta\tb}=4$ we get $4=s_2+2s_3+3s_4+4s_5$. The $s_i$ are not uniquely determined. If $k_{\ta^2\tb}=2$ is also given, we obtain the equation $2=s_3+3s_4+6s_5$ and thus $s_3=2$ and $s_4=s_5=0$ is the only solution. Substituting these results in the previous equation leads to $s_2=0$ and since we only have six $\tb$, we get $s_1=4$.
Hence $w=\tb^4\ta^2\tb^2\ta^2$ is uniquely reconstructed by $S=\{\ta,\ta\tb,\ta^2\tb\}$.

\medskip

The following definition captures all solutions for the equation defined by $k_{\ta^{\ell}\tb}$ for $\ell\in[k_\ta]_0$.

\begin{definition}
Set 
$M(k_{\ta^{\ell}\tb})=\{(r_{\ell+1},\dots,r_{k_\ta +1})|\,k_{\ta^{\ell}\tb}=\sum_{i=\ell+1}^{k_{\ta}+1}\binom{i-1}{\ell}r_i\}$ for fixed $\ell\in[k_\ta]_0$.
We call $k_{\ta^{\ell}\tb}$ {\em unique} if $|M(k_{\ta^{\ell}\tb})|=1$.
\end{definition}

By Remark~\ref{ci} the coefficients of each equation of the form (\ref{alb}) are strictly 
increasing. The next lemma provides the range each $k_{\ta^{\ell}\tb}$ may take under the 
constraint $\sum_{i=1}^{k_\ta+1}s_i=n-k_{\ta}$.

\begin{lemma}\label{maxvalues}
Let $n\in\N$, $k\in[n]_0$, $j\in[k+1]$ and 
$c_1,\dots,c_{k+1},s_1,\dots,s_{k+1}\in\N_0$ with $c_i<c_{i+1}$, for $i\in 
[k]$, and $\sum_{i=1}^{k+1}s_i=n-k$.  The sum $\sum_{i=j}^{k+1}c_is_i$ is 
maximal iff $s_{k+1}=n-k$ (and 
consequently $s_i=0$ for all $ i\in[k]$).
\end{lemma}

\begin{proof}
The case $k=0$ is trivial. Consider the case $n=k$, i.e., $\sum_{i=1}^{k+1}s_i=0$. This implies 
immediately $s_i=0$ for all $i\in[k+1]$ and the equivalence holds. Assume for 
the rest of the proof $k<n$. 
If $s_{k+1} = n-k$, then $s_i = 0$ for all $i \leq k$ and $\sum_{i=j}^{k+1} c_i s_i = c_{k+1} (n-k)$. 
Let us assume that the maximal value for $\sum_{i=j}^{k+1} c_i s_i$ can be obtained in another way and that there exist $s_1', \ldots, s_{k+1}' \in \mathbb{N}_0$, $\ell \in [n-k]$ such that $\sum_{i=1}^{k+1} s_i' = n-k$ and $s_{k+1}' = n-k-\ell$.  
Thus
\[
c_{k+1}(n-k)\leq
\sum_{i=j}^{k+1}c_i s_i' =\left(\sum_{i=j}^{k}c_i s_i' \right)+c_{k+1}(n-k-\ell).
\]
This implies $\sum_{i=j}^{k}c_i s_i' \geq c_{k+1}\ell$.
Since the coefficients are strictly increasing we get 
$\sum_{i=j}^{k}c_i s_i' \leq c_k\sum_{i=j}^{k}
s_i' <c_{k+1}\ell$, hence the contradiction. \qed
%
%
%
%
\end{proof}

\begin{corollary}
Let $k_{\ta}\in[n]_0$, $\ell\in[k_\ta]_0$, and 
$s_1,\dots,s_{k_{\ta}+1}\in\N_0$ with $\sum_{i=1}^{k_\ta+1}s_i=n-k_\ta$.
Then 
$\binom{w}{\ta^{\ell}\tb}\in\left[\binom{k_\ta}{\ell}(n-k_\ta)\right]_0$.
\end{corollary}

\begin{proof}
It follows directly from Equation~\eqref{alb} and Lemma~\ref{maxvalues}.\qed
\end{proof}

The following lemma shows some cases in which $k_{\ta^{\ell}\tb}$ is unique.

\begin{lemma}\label{uniquekalb}
Let $k_{\ta}\in[n]$, $\ell\in[k_\ta]_0$ and 
$s_1,\dots,s_{k_\ta+1}\in\N_0$ with $\sum_{i=1}^{k_\ta+1}s_i=n-k_\ta$.
 If 
$k_{\ta^{\ell}\tb}\in[\ell]_0\cup\{\binom{k_\ta}{\ell}(n-k_\ta)\}$ or 
$k_{\ta^{\ell}\tb}=\binom{k_\ta-1}{\ell}r+\binom{k_\ta}{\ell}(n-k_{\ta}-r)$ for $r\in[k_\tb]_0$ 
then 
$k_{\ta^{\ell}\tb}$ is unique.
\end{lemma}

\begin{proof}
Consider firstly $k_{\ta^{\ell}\tb}\in[\ell]_0$. By Remark~\ref{ci} we have 
$c_{\ell+1}=1$ and $c_{\ell+2}=\ell+1$. By $c_i<c_{i+1}$ we obtain 
immediately $s_i=0$ for $i\in[k_\ta+1]\backslash[\ell+1]$. By setting 
$s_{\ell+1}=k_{\ta^{\ell}\tb}$ the claim is proven. If 
$k_{\ta^{\ell}\tb}=\binom{k_\ta}{\ell}(n-k_\ta)$, $s_{k_\ta+1}=(n-k_\ta)$ and 
$s_i=0$ for $i\in[k_\ta]_0$ is the only possibility. Let secondly be $r\in[k_\tb]_0$ and 
$k_{\ta^{\ell}\tb}=\binom{k_\ta-1}{\ell}r+\binom{k_\ta}{\ell}(n-k_{\ta}-r)$ and suppose that 
$k_{\ta^{\ell}\tb}$ is not unique. This implies $s_{k_{\ta}+1}<n-k_\ta-r$. Assume that 
$s_{k_\ta+1}=n-k_\ta-r'$ for $r'\in[k_\tb]_{>r}$. Thus there exists $x\in\N$ with 
$\binom{k_\ta}{\ell}(n-k_\ta-r')+x=\frac{(k_\ta-1)!(k_\ta(n-k_\ta)-\ell r)}{\ell!(k_\ta-\ell)!}$, i.e.,
$x=\frac{(k_\ta-1)!(k_\ta r'-\ell r)}{\ell!(k_\ta -\ell)!}$. 
By $k_\tb=n-k_\ta$ we have $x\leq \binom{k_\ta-1}{\ell}r'=\frac{(k_\ta-1)!(k_\ta r'-\ell r')}{\ell!(k_\ta-\ell)!}$ (we only have $r'$ occurrences of $\tb$ left to distribute). By $r'>r$ we have $\frac{(k_\ta-1)!(k_\ta r'-\ell r)}{\ell!(k_\ta-\ell)!}=x<\frac{(k_\ta-1)!(k_\ta r'-\ell r)}{\ell!(k_\ta-\ell)!}$ - a contradiction.\qed
\end{proof}

Since we are not able to fully characterise the uniquely determined values for each 
$k_{\ta^{\ell}\tb}$ for arbitrary $n$ and $\ell$, the following proposition gives the 
characterisation for $\ell\in\{0,1\}$. Notice that we use $k_{\ta}$ immediately since it is 
determinable by $n$ and $k_{\ta^0\tb}=k_{\tb}$.

\begin{proposition}\label{smallell}
The word $w \in \Sigma^{n}$ is uniquely determined by $k_\ta$ and $k_{\ta\tb}$ iff one of the 
following occurs 
\begin{itemize}
\item
$k_\ta = 0$ or $k_\ta = n$ (and obviously $k_{\ta\tb} = 0$),
\item
$k_\ta = 1$ or $k_\ta = n-1$ and $k_{\ta\tb}$ is arbitrary,
\item
$k_\ta \in [n-2]_{\geq 2}$ and $k_{\ta\tb} \in \{0,1,k_\ta(n-k_\ta) - 1, k_\ta(n-k_\ta) \}$.
\end{itemize}
\end{proposition}
\begin{proof}
Let us first prove that $w$ is uniquely determined in these cases. It is obvious if $k_\ta = 0$ or 
$k_\ta = n$ since the word is composed of the same letter repeated $n$ times.
If $k_\ta = 1$, then $w = \tb^{s_1} \ta \tb^{n-1-s_1}$ and $\binom{w}{\ta\tb} = n-1-s_1 = 
k_{\ta\tb}$. Therefore $w$ is uniquely determined.
If $k_\ta = n-1$, then $ w = \tb^{s_1} \ta \tb^{s_2} \cdots \ta \tb^{s_n}$ with exactly one of the 
$s_i$ being non zero and, in fact, equal to one. We have $\binom{w}{\ta\tb} = \sum_{i=2}^{n} (i-1) 
s_i$ and, if $k_{\ta\tb}$ is given (between $0$ and $n-1$), then $s_{k_{\ta\tb}+1} = 1$ is the only 
non zero exponent. Consider now $k_\ta\in[n-2]_{\geq 2}$, i.e. $w=\tb^{s_1}\ta\tb^{s_2}\dots \tb^{s_{k_\ta}}\ta\tb^{s_{k_\ta+1}}$.
 Thus $k_{\ta\tb} = 0$ implies $s_1 = n - k_\ta $ and $s_2 = 0, \ldots, s_{k_\ta + 1} = 0$ while 
$k_{\ta\tb} = 1$ implies $s_2 = 1 $, $s_1 = n - k_\ta - 1$ and $s_3 =0, \ldots, s_{k_\ta + 1} = 0$. By Lemma~\ref{maxvalues}, we know that \eqref{alb} is maximal if and only if $s_{k_\ta +1 } = n - 
k_{\ta}$ and all the other $s_i$ are equal to zero. In that case, the value of the sum equals 
$k_{\ta}(n - k_{\ta})$. Therefore, if $\binom{w}{\ta\tb} = k_{\ta}(n - k_{\ta})$, the word $w$ is 
uniquely determined. Finally, if $k_{\ta\tb} = k_{\ta}(n - k_{\ta}) -1$, we must have $s_{k_\ta + 1} \leq n - k_\ta - 1$. 
If we choose $s_{k_\ta + 1} = n - k_\ta - 1$, it remains that $\sum_{i=1}^{k_\ta} s_i = 1$ and 
$\sum_{i=2}^{k_\ta} (i-1)s_i = k_\ta -1$. We must have $s_{k_\ta} = 1$ and the other ones equal to 
zero. 
In fact, choosing $s_{k_\ta + 1} = n - k_\ta - 1$ is the only possibility: if otherwise $s_{k_\ta + 
1} = n - k_\ta - \ell$ with $\ell > 1$, we obtain that $\sum_{i=2}^{k_\ta} (i-1) s_i \geq \ell k_\ta 
- 1$ with $\sum_{i=1}^{k_\ta} s_i = \ell$. It is easy to check with Lemma~\ref{maxvalues} that these 
conditions are incompatible.

\medskip

We now need to prove that $w$ cannot be uniquely determined if $k_\ta \in [n-2]_{\geq 2}$ and 
$k_{\ta\tb} \in [k_\ta (n-k_{\ta}) - 2]_{\geq 2}$. To this aim we will give two different sets of 
values for the $s_i$. The first decomposition is the greedy one. Let us put $s_{k_\ta + 1} = \lfloor 
\frac{k_{\ta\tb}}{k_\ta} \rfloor$, $s_{(k_{\ta\tb} \bmod k_\ta) + 1} = 1$ and the other $s_i$ 
equal to $0$. Let us finally modify the value of $s_1$ (which is, at this stage, equal to $0$ or $1$) by adding the value needed. By $\sum_{i=1}^{k_{\ta}+1} s_i=n-k_\ta$ we get $s_1 \leftarrow s_1 + (n-k_\ta) -\lfloor 
\frac{k_{\ta\tb}}{k_\ta} \rfloor - 1$. 
This implies $\sum_{i=1}^{k_\ta + 1} s_i = 1 + (n-k_\ta) - \left \lfloor \frac{k_{\ta\tb}}{k_\ta} \right \rfloor - 
1 + \left \lfloor \frac{k_{\ta\tb}}{k_\ta} \right \rfloor = n-k_\ta$
and $s_i \geq 0$ for all $i$. Moreover we have $
\sum_{i=2}^{k_\ta + 1} (i-1) s_i = (k_{\ta\tb} \bmod k_\ta) + k_\ta \left \lfloor 
\frac{k_{\ta\tb}}{k_\ta} \right \rfloor = k_{\ta\tb}$.

\medskip

Now we provide a second decomposition for the $s_i$. First, let us assume that $2 \leq k_{\ta\tb} < k_{\ta}$. In that case, the greedy algorithm sets 
$s_{k_{\ta\tb} + 1} = 1$, $s_1 = n - k_\ta - 1$ and the other $s_i$ to $0$. Let us 
now set $s_1 = n - k_{\ta} - 2$ and all the other $s_i$ to $0$. Then, update 
$s_{k_{\ta\tb}} \leftarrow s_{k_{\ta\tb}} + 1$ and $s_2 \leftarrow s_2 + 1$ (in the case where $k_{\ta\tb} = 2$, $s_2$ will be equal to $2$ after these manipulations).
We have that the sum in \eqref{alb} is equal to $1 + (k_{\ta\tb} - 1)$ as needed.
Finally, if $k_{\ta\tb} \geq k_\ta$, then $s_{k_\ta + 1}$ was non zero in the greedy decomposition, 
and the idea is to reduce it of a value $1$. 
Let us set $s_{k_\ta + 1} =\lfloor \frac{k_{\ta\tb}}{k_\ta} \rfloor - 1$ and the other $s_i$ to $0$. 
Then, let us update some values: $s_{(k_{\ta\tb} \bmod k_\ta) + 2} \leftarrow s_{(k_{\ta\tb} \bmod k_\ta) + 2} + 1$ and $s_{k_\ta} \leftarrow s_{k_\ta} + 1$ if 
$(k_{\ta\tb} \bmod k_\ta) \neq k_\ta - 1$, and $s_{k_\ta} = 2$, $s_2 = 1$ otherwise. Finally, set 
$s_1$ to the right value, i.e., $n-k_\ta - \sum_{i=2}^{k_\ta +1} s_i$.
It can be easily checked that, in both cases, $s_1 \geq 0$ (notice that $(k_{\ta\tb} \bmod k_\ta) = 
k_\ta - 1$ implies that $\lfloor \frac{k_{\ta\tb}}{k_\ta} \rfloor \leq n - k_\ta -2$) and that all 
$s_i$ sum up to $n - k_\ta$. 
Similarly, we can check that $\sum_{i=2}^{k_\ta + 1} (i-1) s_i$ is equal to $k_{\ta\tb}$ in both 
cases.

\medskip

To sum up, we gave two different decompositions for the $s_i$ in cases where $k_\ta \in [n-2]_{\geq 
2}$ and $k_{\ta\tb} \in [k_\ta (n-k_{\ta}) - 2]_{\geq 2}$. That implies that $w$ cannot be uniquely 
determined in those cases.\qed
\end{proof}

In all cases not covered by Proposition~\ref{smallell} the word cannot be uniquely 
determined by $\binom{w}{\ta}$ and $\binom{w}{\ta\tb}$. The following theorem combines the 
reconstruction of a word with the binomial coefficients of right-bounded-block words.

\begin{theorem}\label{uniquerecon}
Let $j\in[k_\ta]_0$.
If $k_{\ta^{j}\tb}$ is unique, then the word $w\in\Sigma^n$ is uniquely determined by $\{\tb,\ta\tb,\ta^2\tb,\dots,\ta^{j}\tb\}$.
\end{theorem}

\begin{proof}
If $k_{\ta^{j}\tb}$ is unique, the coefficients 
$s_{j+1},\dots,s_{k_{\ta}+1}$ are uniquely determined. Substituting backwards the known values in the first $j-1$ equations (\ref{alb}) (for $\ell = 1, \ldots, j-1$)  we can now obtain successively the values for $s_{j},\dots,s_1$.\qed
\end{proof}

\begin{corollary}
Let $\ell$ be minimal such that $k_{\ta^{\ell}\tb}$ is unique. Then $w$ is uniquely determined by $\{\ta,\ta\tb,\ta^2\tb,\dots,\ta^{\ell}\tb\}$ and not uniquely determined by any \linebreak $\{\ta,\ta\tb,\ta^2\tb,\dots,\ta^{i}\tb\}$ for $i<\ell$.
\end{corollary}

\begin{proof}
It follows directly from Theorem~\ref{uniquerecon}.\qed
\end{proof}

By \cite{KrRo} an upper bound on the number of binomial coefficients to uniquely reconstruct the word 
$w\in\Sigma^n$ is given by the amount of the binomial coefficients of the $(\lfloor \frac{16}{7}\sqrt{n}\rfloor +5)$-spectrum. Notice that implicitly the full spectrum is assumed to be known.
As proven in Section~\ref{prel}, Lyndon words up to this length suffice. 
Since there are $\frac{1}{n}\sum_{d|n}\mu(d)\cdot 2^{\frac{n}{d}}$ Lyndon words of length $n$,
the combination of both results presented in \cite{KrRo,Reut} states that, for $n > 6$,
\begin{align}\label{eqform}
\sum_{i=1}^{\lfloor 
\frac{16}{7}\sqrt{n}\rfloor +5} \frac{1}{i}\sum_{d|i}\mu(d)\cdot 2^{\frac{i}{d}}
\end{align}
binomial coefficients are sufficient for a unique reconstruction with the Möbius function $\mu$. Up to now, it was the best known upper bound.

Theorem~\ref{uniquerecon} shows that $\min\{k_\ta,k_\tb\}+1$ binomial coefficients are enough for reconstructing a binary word uniquely. By Proposition~\ref{smallell} we need exactly one binomial coefficient if $n\in[3]$ and at most two if $n=4$. For $n\in\{5,6\}$ we need at most $n-2$ different binomial coefficients.
The following theorem shows that by Theorem~\ref{uniquerecon}  we need strictly less binomial coefficients for $n > 6$.

\begin{theorem}\label{binarybound}
Let $w\in\Sigma^n$. We have that $\min\{k_\ta,k_\tb\}+1$ binomial coefficients suffice to uniquely reconstruct $w$. If $k_{\ta} \leq k_{\tb}$, then the set of sufficient binomial coefficients is $S = \{\tb,\ta\tb,\ta^2\tb,...,\ta^h \tb\}$ where $h=\lfloor\frac{n}{2}\rfloor$. If $k_{\ta} > k_{\tb}$, then the set is $S = \{\ta,\tb\ta,\tb^2\ta,...,\tb^h \ta\}$.
This bound is strictly smaller than (\ref{eqform}).
\end{theorem}

\begin{proof}
Assume w.l.o.g. $k_\ta\leq k_\tb$.
Then $k_{\ta} \leq \frac{n}{2}$ and Theorem~\ref{uniquerecon} shows that words in the set $ \{ \tb, \ta \tb, \ldots, \ta^{\lfloor \frac{n}{2} \rfloor} \tb \}$ can reconstruct $w$ uniquely. If $k_{\ta} > k_{\tb}$, the set $S$ is obtained by replacing the letter $\ta$ by $\tb$ and vice-versa.

Set 
$N_2(i):=\frac{1}{i}\sum_{d|i}\mu(d)2^{\frac{i}{d}}$ for all $i\in[\lfloor\frac{16}{7}\sqrt{n}\rfloor+5]$, i.e.,
$\sum_{i=1}^{\lfloor 
\frac{16}{7}\sqrt{n}\rfloor +5} N_2(i)$, which is Equation~\eqref{eqform}, binomial coefficients suffice. By 
\cite[Lemma~2.4]{Ferov} we have
\begin{align*}
N_2(i) & \geq 
\frac{1}{i}\left(2^i-\frac{2^{\frac{i}{2}}-1}{2-1}\right)
= \frac{1}{i}\left(2^i-2^{\frac{i}{2}}+1\right)
= \frac{1}{i}\left(2^{\frac{i}{2}} (2^{\frac{i}{2}}-1)+1\right)
\geq \frac{2^{\frac{i}{2}}}{i}.
\end{align*}
This results in 
\begin{align*}
\sum_{i=1}^{\lfloor 
\frac{16}{7}\sqrt{n}\rfloor +5} N_2(i) &\geq \sum_{i=1}^{\lfloor 
\frac{16}{7}\sqrt{n}\rfloor +5} \frac{2^{\frac{i}{2}}}{i}
\geq 
\frac{1}{\frac{16}{7}\sqrt{n} +5} \, \frac{\sqrt{2}^{\frac{16}{7} \sqrt{n}+5}-1}{\sqrt{2}-1}.
\end{align*}
We want to show that this quantity is at least equal to $\frac{n+1}{2}$. Let us define 
$$
f(x) = \frac{1}{\frac{16}{7}\sqrt{x} +5} \, \frac{\sqrt{2}^{\frac{16}{7} \sqrt{x}+5}-1}{\sqrt{2}-1} 
- \frac{x+1}{2}
$$
for all $x > 0$, which is the continuous extension on $\mathbb{R}^{+}$ of the quantity we are 
interested in. It is easy to verify by hand that $f(1), f(2), f(3)$ and $f(4)$ are positive. Let us 
formally show that $f(x) > 0$ for all $x \geq 5$. Since this function is differentiable, we get 
with $y = \frac{16}{7} \sqrt{x} + 5$
\[
f'(x) = \frac{1}{y} \, \sqrt{2}^{y} \, 
\frac{\ln(\sqrt{2})}{\sqrt{2}-1} \, \frac{8}{7 \sqrt{x}} - \frac{1}{y^2} 
\, \frac{8}{7 \sqrt{x}} \, \frac{\sqrt{2}^{y}-1}{\sqrt{2}-1} - \frac{1}{2}.
\]
We thus have $\sqrt{x} = \frac{7y-35}{16}$ and $ y 
\geq 7$ for all $x \geq 1$.  By injecting $y$ in the previous expression, and reducing to the common 
denominator, we have to show that
\begin{align*}
 & \, 2y \sqrt{2}^y 128 \ln(\sqrt{2}) - 256 (\sqrt{2}^y - 1) - 7(7y-35) y^2(\sqrt{2}-1)  \\
 = & \,\sqrt{2}^y (128 \ln(2) y - 256) + 256 - 49y^3 (\sqrt{2} - 1) + 245 y^2 (\sqrt{2} - 1) \\
 \geq & \, \sqrt{2}^y 365 + 256 - 49y^3 (\sqrt{2} - 1) + 245 y^2 (\sqrt{2} - 1)
\end{align*}
is strictly positive. Let us call the last quantity $g(y)$. We will show that it is positive for 
all $y \geq 10.05$, which means that $f(x)$ is positive for all $x$ such that $\frac{16}{7} \sqrt{x} 
+ 5 \geq 10.05$, i.e., for all $x \geq 5$. 
We have 
\begin{align*}
g'(y) &= 365 \sqrt{2}^y \ln(\sqrt{2}) - 147(\sqrt{2}-1) y^2 + 490(\sqrt{2}-1)y, \\
g''(y) &= 365 \sqrt{2}^y (\ln(\sqrt{2}))^2 - 294(\sqrt{2}-1) y + 490(\sqrt{2}-1), \\
g'''(y) &= 365 \sqrt{2}^y (\ln(\sqrt{2}))^3 - 294(\sqrt{2}-1), 
\end{align*}
and $g'''(7) > 50$, $g''(8.5) > 2$, $g'(10.05) > 8$ and finally $g(10.05) > 1787$. 
Since $g'''(y)$ is increasing and positive in $7$, $g''(y)$ is increasing for $y \geq 7$. Therefore 
$g'(y)$ is increasing for $y \geq 8.5$ and finally $g(y)$ is increasing for $y \geq 10.05$ and 
positive.\qed
\end{proof}

\begin{remark}
By Lemma~\ref{uniquekalb} we know that $k_{\ta^{\ell}\tb}$ is unique if it is in $[\ell]_0$ or exactly $\binom{k_\ta}{\ell}(n-k_\ta)$. The probability for the latter is $\frac{1}{2^n}$ for $w\in\{\ta,\tb\}^n$. If $k_{\ta^{\ell}\tb}=m\in[\ell]_0$ we get by (\ref{alb}) immediately $s_{\ell+1}=m$ and $s_{i}=0$ for $\ell+2\leq i\leq k_\ta+1$. Hence, the values for $s_j$ for $j\in[\ell]$ are not determined. By $\sum_{i\in[\ell]}s_i=n-k_\ta-m$ there are $d=\sum_{i\in[\ell]_0}\binom{\ell}{\ell-i}\binom{n-k_\ta-m-1}{i-1}$ possibilities to fulfill the constraints, i.e., we have a probability of $\frac{d}{2^n}$ to have such a word.
\end{remark}

\section{Reconstruction for Arbitrary Alphabets}\label{general}
In this section we address the problem of reconstructing words over arbitrary alphabets from their scattered factors. We begin with a series of results of algorithmic nature. Let $\Sigma=\{\ta_1,\dots,\ta_q\}$ be an alphabet equipped with the ordering $\ta_i<\ta_j$ for $1\leq i<j\leq q\in\N$.

\begin{definition}
Let $w_1,\ldots, w_k\in\Sigma^{\ast}$ for $k\in\N$, and $K=(k_\ta)_{\ta\in \Sigma}$ a sequence of $|\Sigma|$ natural numbers. 
 A {\em $K-$valid marking} of  $w_1,\ldots, w_k$ is a mapping $\psi: [k]\times \N \rightarrow \N$ such that for all $j\in [k]$, $i,\ell\in [|w_j|]$, and $\ta\in \Sigma$ there holds
\begin{itemize}
\item if $w_j[i]=\ta$ then $\psi(j,i)\leq k_\ta$,
\item if $i<\ell\leq |w_j|$ and $w_j[i]=w_j[\ell]=\ta$ then $\psi(j,i)<\psi(j,\ell)$.
\end{itemize}
A {\em $K$-valid marking} of  $w_1,\ldots, w_k$ is represented as the string $w_1^\psi$, $w_2^\psi, \ldots$, $w_k^\psi$, where $w_j^\psi [i]=(w_j[i])_{\psi(j,i)}$ for fresh letters $(w_j[i])_{\psi(j,i)}$. 
\end{definition}

For instance, let $k=2$, $\Sigma=\{\ta,\tb\}$, and $w_1=\ta\ta\tb$, $w_2=\ta\tb\tb$. Let $k_\ta=3, k_\tb=2$ define the sequence $K$. A $K$-valid marking of $w_1, w_2$ would be $w_1^\psi=(\ta)_1(\ta)_3(\tb)_1, w_2^\psi =(\ta)_2 (\tb)_1(\tb)_2$ defining $\psi$ implicitly by the indices. We used parentheses in the marking of the letters in order to avoid confusions.

We recall that a topological sorting of a directed graph $G=(V,E)$, with $V=\{v_1,\ldots,v_n\}$, is a linear ordering $v_{\sigma(1)} <  v_{\sigma(2)} < \ldots < v_{\sigma(n)}$ of the nodes, defined by the permutation $\sigma:[n]\rightarrow [n]$, such that there exists no edge in $E$ from $v_{\sigma(i)}$ to $v_{\sigma(j)}$ for any $i>j$ (i.e., if $v_{a}$ comes after $v_b$ in the linear ordering, for some $a=\sigma(i)$ and $b=\sigma(j)$, then we have $i>j$ and there should be no edge between $v_a$ and $v_b$). It is a folklore result that any directed graph $G$ has a topological sorting if and only if $G$ is acyclic.

\begin{definition}
Let $w_1,\ldots, w_k\in\Sigma^{\ast}$ for $k\in\N$, $K=(k_\ta)_{\ta\in \Sigma}$ a sequence of $|\Sigma|$ natural numbers, and $\psi $ a {\em $K-$valid marking} of  $w_1,\ldots, w_k$. Let $G_\psi$ be the graph that has $\sum_{\ta\in \Sigma} k_\ta$ nodes, labelled with the letters $(\ta)_1,\ldots,(\ta)_{k_\ta}$, for all $\ta\in \Sigma$,  and the directed edges $((w_j[i])_{\psi(j,i)},(w_j[i+1])_{\psi(j,i+1)})$, for all $j\in[k]$, $i\in[|w_j|]$, and $((\ta)_i,(\ta)_{i+1})$, for all occuring $i$ and $\ta\in\Sigma$. 
We say that there exists {\em a valid topological sorting} of the $\psi$-marked letters of the words $w_1,\ldots,w_k$ if there exists a topological sorting of the nodes of $G_{\psi}$, i.e., $G_{\psi}$ is a directed acyclic graph.
\end{definition}

The graph associated with the $K$-valid marking of $w_1, w_2$ from above would have the five nodes $(\ta)_1,(\ta)_2,(\ta)_3,(\tb)_1,(\tb)_2$ and the six directed edges $((\ta)_1,(\ta)_3)$, $((\ta)_3,(\tb)_1)$, $((\ta)_2,(\tb)_1)$, $((\tb)_1,(\tb)_2)$, $((\ta)_1,(\ta)_2), ((\ta)_2,(\ta)_3)$ (where the direction of the edge is from the left node to the right node of the pair defining it). This graph has the topological sorting $(\ta)_1(\ta)_2(\ta)_3(\tb)_1(\tb)_2$. 

\begin{theorem}\label{recon1}
For $w_1, \ldots, w_k\in \Sigma^*$ and a sequence $K=(k_\ta)_{\ta\in \Sigma}$ of $|\Sigma|$ natural numbers, there exists a word $w$ such that $w_i$ is a scattered factor of $w$ with $|w|_{\ta}=k_\ta$, for all $i\in [k]$ and all $\ta\in\Sigma$, 
if and only if there exist a $K$-valid marking $\psi$ of the words $w_1,\ldots,w_k$ and a valid topological sorting of the $\psi$-marked letters of the words $w_1,\ldots,w_k$. 
\end{theorem}

\begin{proof}
If $w$ is such that $w_i$ is a scattered factor of $w$, for all $i\in [k]$, and $|w|_\ta=k_\ta$, for all $\ta\in \Sigma$, then we can mark the $i^{th}$ occurrence of $\ta$ as $(\ta)_i$, for all $\ta\in \Sigma$ and $i\in [k_\ta]$. This induces a $K$-valid marking $\psi$ of the words $w_i$, and, moreover, the linear ordering of the nodes of $G_\psi$ induced by the order in which the marked letters (i.e., nodes of $G_\psi$) occur in $w$ is a topological sorting of $G_{\psi}$. 

Let us now assume that there exists a $K$-valid marking $\psi$ of the words $w_1,\ldots,w_k$, and there exists a valid topological sorting of the $\psi$-marked letters of the words $w_1,\ldots,w_k$. Let $w'$ be the word obtained by writing the nodes of $G_{\psi}$ in the order given by its topological sorting and removing their markings. It is clear that $w'$ has $w_i$ as a scattered factor, for all $i\in [k]$, and that $|w'|_\ta\leq k_\ta$, for all $\ta\in \Sigma$. Let now $w=w' \prod_{\ta\in \Sigma}\ta^{k_\ta-|w'|_\ta}$, where $\prod_{\ta\in \Sigma}\ta^{k_\ta-|w'|_\ta} $ is the concatenation of the factors $\ta^{k_\ta-|w'|_\ta}$, for $\ta\in \Sigma$ in some fixed order. Now $w$ has $w_i$ as a scattered factor, for all $i\in [k]$, and $|w|_\ta= k_\ta$, for all $\ta\in \Sigma$.
\qed\end{proof}

Next we show that in Theorem~\ref{recon1} uniqueness propagates in the $\Leftarrow$-direction.

 \begin{corollary}\label{reconst_topo}
Let $w_1, \ldots, w_k\in \Sigma^*$ and $K=(k_\ta)_{\ta\in \Sigma}$ a sequence of $|\Sigma|$ natural numbers. If the following hold
\begin{itemize}
\item there exists a unique $K$-valid marking $\psi$ of the words $w_1,\ldots,w_k$,
\item in the unique $K$-valid marking $\psi$ we have that for each $\ta\in \Sigma$ and $\ell \in [k_\ta]$ there exists $i\in [k]$ and $j\in [|w_i|]$ with $\psi(i,j)=\ell$, and
\item there exists a unique valid topological sorting of the $\psi$-marked letters of the words $w_1,\ldots,w_k$
\end{itemize}
then there exists a unique word $w$ such that $w_i$ is a scattered factor of $w$, for all $i\in [k]$ and $|w|_\ta=k_\ta$ for all $\ta\in \Sigma$.
\end{corollary}

\begin{proof}
Let $w$ be the word obtained by writing in order the letters of the unique valid topological sorting of the $\psi$-marked letters of the words $w_1,\ldots,w_k$ and removing their markings. It is clear that $w'$ has $w_i$ as a scattered factor, for all $i\in [k]$, and that $|w|_\ta= k_\ta$, for all $\ta\in \Sigma$. The word $w$ is uniquely defined (as there is no other $K$-valid marking nor valid topological sorting of the $\psi$-marked letters), and $|w|_\ta=k_\ta$, for all $\ta\in \Sigma$. 
\qed\end{proof}

In order to state the second result, we need the projection $\pi_S(w)$ of a word $w\in\Sigma^{\ast}$ on $S\subseteq\Sigma$: $\pi_S(w)$ is obtained from $w$ by removing all letters from $\Sigma \setminus S$.

\begin{theorem}\label{recon2}
Set $W=\{w_{\ta,\tb}\mid \ta<\tb\in \Sigma \}$ such that 
\begin{itemize}
\item $w_{\ta,\tb}\in \{\ta,\tb\}^*$ for all $\ta,\tb\in \Sigma$,
\item for all $w,w'\in W$ and all $\ta \in \Sigma$, if $|w|_{\ta} \cdot |w'|_{\ta} > 0$, then $|w|_{\ta} = |w'|_{\ta}$.
\end{itemize}
Then there exists at most one $w\in \Sigma^*$ such that $w_{\ta,\tb}$ is $\pi_{\{\ta,\tb\}}(w)$ for all $\ta,\tb\in \Sigma$.
\end{theorem}
\begin{proof}
Notice firstly $|W|=\frac{q(q-1)}{2}$.
Let $k_\ta=|w_{\ta,\tb}|_\ta$, for $\ta<\tb\in \Sigma$. These numbers are clearly well defined, by the second item in our hypothesis. Let $K=(k_\ta)_{\ta\in \Sigma}$. It is immediate that there exists a unique $K$-valid marking $\psi$ of the words $(w_{\ta,\tb})_{\ta<\tb\in \Sigma}$. As each two marked letters $(\ta)_i$ and $(\tb)_j$ (i.e., each two nodes $(\ta)_i$ and $(\tb)_j$ of $G_\psi$) appear in the marked word $w^\psi_{\ta,\tb}$, we know the order in which these two nodes should occur in a topological sorting of $G_{\psi}$. This means that, if $G_\psi$ is acyclic, then it has a unique topological sorting. Our statement follows now from Corollary \ref{reconst_topo}.
\qed\end{proof}

\begin{remark}\label{algo}
Given the set $W=\{w_{\ta,\tb}\mid \ta<\tb\in \Sigma \}$ as in the statement of Theorem \ref{recon2}, with $k_\ta=|w_{\ta,\tb}|_\ta$, for $\ta<\tb\in \Sigma$, and $K=(k_\ta)_{\ta\in \Sigma}$, we can produce the unique $K$-valid marking $\psi$ of the words  $(w_{\ta,\tb})_{\ta<\tb\in \Sigma}$ in linear time $O(\sum_{\ta<\tb\in \Sigma} |w_{\ta,\tb}|)=O((q-1)\sum_{\ta \in \Sigma} k_\ta)$: just replace the \nth{$i$} letter $\ta$ of $w_{\ta,\tb}$ by $(\ta)_i$, for all $\ta$ and $i$. The graph $G_{\psi}$ has $O((q-1)\sum k_{\ta})$ edges and $O(\sum k_\ta)$ vertices and can be constructed in linear time $O((q-1)\sum k_\ta)$. Sorting $G_\psi$ topologically takes $O((q-1)\sum k_\ta)$ time (see, e.g., the handbook \cite{cormen}). As such, we conclude that reconstructing a word $w\in \Sigma^*$ from its projections over all two-letter-subsets of $\Sigma$ can be done in linear time w.r.t. the total length of the respective projections. 
\end{remark}

Theorem \ref{recon2} is in a sense optimal: in order to reconstruct a word over $\Sigma$ uniquely, we need all its projections on two-letter-subsets of $\Sigma$. That is, it is always the case that for a strict subset $U$ of $\{\{\ta,\tb\}\mid \ta<\tb\in \Sigma\}$, with $|U|=\frac{q(q-1)}{2} -1$, there exist two words $w'\neq w$ such that $\{\pi_p(w')\mid p\in U\}=\{\pi_p(w)\mid p\in U\}$. We can, in fact, show the following results:
\begin{theorem}\label{char_proj}
Let $S_1,\ldots, S_k$ be subsets of $\Sigma$. The following hold:
\begin{enumerate}
\item If each pair $\{\ta,\tb\}\subseteq \Sigma$ is included in at least one of the sets $S_i$, then we can reconstruct any word uniquely from its projections $\pi_{S_1}(\cdot), \ldots, \pi_{S_k}(\cdot)$. 
\item If there exists a pair $\{\ta,\tb\}$ that is not contained in any of the sets $S_1,\ldots, S_k$, then there exist two words $w$ and $w'$ such that $w\neq w'$ and $\pi_{S_1}(w)=\pi_{S_1}(w'), \ldots,$ $\pi_{S_k}(w)=\pi_{S_k}(w')$. 
\end{enumerate}
\end{theorem}

\begin{proof}
The first part is, once again, a consequence of Corollary \ref{reconst_topo}. The second part can be shown by assuming that $\Sigma=\{\ta_1,\ldots,\ta_q\}$ and the pair $\{\ta_1,\ta_2\}$ is not contained in any of the sets $S_1,\ldots,S_k$. Then, for $w=\ta_1\ta_3 \ta_4\ldots \ta_q$ and $w'=\ta_2\ta_3 \ta_4\ldots \ta_q$, we have that $\pi_{S_1}(w)=\pi_{S_1}(w'), \ldots,$ $\pi_{S_k}(w)=\pi_{S_k}(w')$. 
\qed
\end{proof}

In this context, we can ask
how efficiently can we decide if a word is uniquely reconstructible from the projections $\pi_{S_1}(\cdot), \ldots$, $\pi_{S_k}(\cdot)$ for $S_1,\dots,S_k\subset\Sigma$.

\begin{theorem}\label{complexity}
Given the sets $S_1,\ldots, S_k\subset \Sigma$, we decide whether we can reconstruct any word uniquely from its projections $\pi_{S_1}(\cdot), \ldots, \pi_{S_k}(\cdot)$ in $O(q^2 k)$ time. Moreover, under the {\em Strong Exponential Time Hypothesis} (see the survey~\cite{bringmann} and the references therein), there is no $O(q^{2-d} k^c)$ algorithm for solving the above decision problem, for any $d,c>0$.
\end{theorem}

\begin{proof}
We begin with a series of preliminaries. Let us recall the {\em Orthogonal Vectors} problem: Given sets $A, B$ consisting of $n$ vectors in $\{0,1\}^k$, decide whether there are vectors $a\in A$ and $ b\in B$ which are orthogonal (i.e., for any $i\in [k]$ we have $a[i]b[i] = 0$). This problem can be solved na\"ively in $O(n^2k)$ time, but under the {\em Strong Exponential Time Hypothesis} there is no $O(n^{2-d} k^c)$ algorithm for solving it, for any $d,c>0$ (once more, see the survey \cite{bringmann} and the references therein).

We show that our problem is equivalent to the {\em Orthogonal Vectors} problem. 

Let us first assume that we are given the sets $S_1,\ldots, S_k\subset \Sigma$, and we want to decide whether we can reconstruct any word uniquely from its projections $\pi_{S_1}(\cdot), \ldots, \pi_{S_k}(\cdot)$. This is equivalent, according to Theorem \ref{char_proj}, to checking whether each pair $\{\ta,\tb\}\subseteq \Sigma$ is included in at least one of the sets $S_i$. For each letter $\ta$ of $\Sigma$ we define the $k$-dimensional vectors $x_\ta$ where $x_\ta[i]=1$ if $\ta \in S_i$ and $x_\ta[i]=0$ if $\ta \not\in S_i$. This can be clearly done in $O(q k)$ time. Now, there exists a pair $\{\ta,\tb\}$ that is not contained in any of the sets $S_1,\ldots, S_k$ if and only if there exists a pair of vectors $\{x_\ta,x_\tb\}$ such that $x_\ta[i] x_\tb[i]=0$ for all $i\in [k]$. We can check whether there exists a pair of vectors $\{x_\ta,x_\tb\}$ such that $x_\ta[i] x_\tb[i]=0$ for all $i\in [k]$ by solving the {\em Orthogonal Vectors} problem by using for both input sets of vectors the set $\{x_\ta\mid \ta \in \Sigma\}$. As such, we can check whether there exists a pair $\{\ta,\tb\}$ that is not contained in any of the sets $S_1,\ldots, S_k$ in $O(q^2 k)$ time.

Let us now assume that we are given two sets $A, B$ consisting of $n$ vectors in $\{0,1\}^k$, and we want to decide whether there are vectors $a\in A$ and $ b\in B$ which are orthogonal (i.e., for any $i\in [k]$ we have $a[i]b[i] = 0$). We can compute the set of $(k+2)$-dimensional vectors $A'$ containing the vectors of $A$ extended with two new positions (position $k+1$ and position $k+2$) set to $10$ and the vectors of $B$ extended with two new positions (position $k+1$ and position $k+2$) set to $01$. To decide whether there are vectors $a\in A$ and $ b\in B$ which are orthogonal is equivalent to decide whether there are vectors $a,b\in A'$ which are orthogonal (if two such vectors exist, they must be different on their last two positions, so one must come from $A$ and one from $B$). Assume that $A'=\{x_1,x_2,\ldots,x_{2n}\}$. Now we define an alphabet $\Sigma=\{\ta_1,\ldots,\ta_{2n}\}$ of size $2n$ and the sets $S_1,\ldots, S_k$, where $\ta_j\in S_i$ if and only if $x_j[i]=1$. Computing $A'$ and then the alphabet $\Sigma$ and the sets $S_i$, for $i\in [k]$, takes $O(nk)$ time. Now, to decide whether there are vectors $x_i,x_j\in A'$ which are orthogonal is equivalent to decide whether there exists a pair of letters $\{\ta_i, \ta_j\}$ of $\Sigma$ that is not contained in any of the sets $S_1,\ldots, S_k$. 
The conclusion of the theorem now follows.
\qed\end{proof}


Coming now back to combinatorial results, we use the method developed in Section~\ref{binary} to reconstruct a word over an arbitrary alphabet. We show that we need at most $\sum_{i\in[q]}|w|_i(q+1-i)$ different binomial coefficients 
to reconstruct $w$ uniquely for the alphabet $\Sigma=\{1,\dots,q\}$. In fact, following the results from the first part of this section, we apply this method on all combinations 
of two letters. Consider for an example that for $w\in\{\ta,\tb,\tn\}^{6}$ the following binomial 
coefficients $\binom{w}{\ta^0\tb}=1$, $\binom{w}{\ta^0\tn}=2$, 
$\binom{w}{\ta^1\tb}=0$, $\binom{w}{\ta^1\tn}=3$, $\binom{w}{\tb^1\tn}=2$, and $\binom{w}{\ta^2\tn}=1$
are given.
By $|w|=6$, $|w|_\tb=1$, and $|w|_\tn=2$, we get $|w|_\ta=3$. Applying the method from Section~\ref{binary} for $\{\ta,\tb\}$, $\{\ta,\tn\}$, and $\{\tb,\tn\}$ we obtain the scattered factors $\tb\ta^3$, $\ta\tn\ta\tn\ta$, and $\tb\tn^2$. Combining all these three scattered factors gives us uniquely 
$\tb\ta\tn\ta\tn\ta$. Notice that in this example we only needed six binomial coefficients instead 
of ten, which is the worst case.

\begin{remark}
As seen in the example we have not only the word length but also $\binom{w}{\tx}$ for all 
$\tx\in\Sigma$ but one. Both information give us the remaining single letter binomial coefficient 
and hence we will assume that we know all of them. 
\end{remark}

For convenience in the following theorem consider $\Sigma=\{1,\dots,q\}$ for $q>2$ and set $\alpha:=\lfloor\frac{16}{7}\sqrt{n}\rfloor +5$. In the general case the results by \cite{Reut} and \cite{KrRo} yield that
\begin{align}\label{generalnumber}
\sum_{i\in[\alpha]}\frac{1}{i}\frac{(q+1)^{\frac{i}{2}}-1}{q}
\end{align}
is smaller than the best known upper bound on the number of binomial coefficients sufficient to reconstruct a word uniquely.

The following theorem generalises Theorem~\ref{binarybound} on an arbitrary alphabet.

\begin{theorem}\label{reconstructgeneral}
For uniquely reconstructing a word $w\in\Sigma^{\ast}$ of length at least \linebreak $q-1$, $\sum_{i\in[q]}|w|_i(q+1-i)$ binomial coefficients suffice, which is strictly smaller than (\ref{generalnumber}).
\end{theorem}

\begin{proof}
The claim that $\sum_{i\in[q]}|w|_i(q+1-i)$ binomial coefficients suffice to reconstruct $w$ uniquely follows by Theorem~\ref{recon2}: for each pair of letters we apply the method of the binary case. 
We are thus going to reconstruct words $w_{\ta,\tb}$ for all pairs of letters $\ta < \tb$. If $\ta$ is the $i^\text{th}$ letter in the alphabet, there are $q-i$ such pairs. To determine $w_{\ta,\tb}$ uniquely,  $\min(k_{\ta}, k_{\tb}) + 1 \leq k_{\ta} + 1$ binomial coefficients from the set $\{k_{\tb}, k_{\ta \tb} ,\ldots, k_{\ta^{|w|_\ta} \tb}  \}$ suffice. In total, we thus need the binomial coefficients of the set
$$
 \{ k_{\ta^j \tb} : \ta < \tb, j \in [|w|_{\ta}] \} \cup \{k_\tb : \tb \in \Sigma \backslash \{1\} \}.
$$ 
There are $
\sum_{i \in [q]} |w|_i (q-i) + (q-1)
$
such coefficients.
This quantity is less than or equal to $\sum_{i \in [q]} |w|_i(q+1-i)$ for every $w$ of length at least $q-1$. 

We show the second claim about the bound by induction on $q$ where the binary case in Theorem~\ref{binarybound} serves as induction basis. This implies
\begin{align*}
\sum_{i\in[q]}|w|_i(q+1-i)
&=\left(\sum_{i\in[q-1]}|w|_i(q+1-i)\right)+|w|_q(q+1-q)\\
&=\left(\sum_{i\in[q-1]}|w|_i(q-i)\right)+\sum_{i\in[q-1]}|w|_i+|w|_q
\end{align*}
and therefore
\begin{align*}
\sum_{i\in[q]}|w|_i(q+1-i) &\leq \left(\sum_{i\in[\alpha]}\frac{1}{i}\frac{q^{\frac{i}{2}}-1}{q-1}\right)+n\\
&=\left(\sum_{i\in[\alpha]}\frac{1}{i}\frac{q(q^{\frac{i}{2}}-1)}{q(q-1)}\right)+n.
\end{align*}
On the other hand, we have to compare this quantity with \eqref{generalnumber} which can be rewritten as 
\[
\sum_{i\in[\alpha]}\frac{1}{i}\frac{(q+1)^{\frac{i}{2}}-1}{q}
=\sum_{i\in[\alpha]}\frac{1}{i}\frac{(q-1)((q+1)^{\frac{i}{2}}-1)}{q(q-1)}.
\]
Thus the claim is proven, if the substraction of the latter one and the previous one is greater than zero, i.e., we show that
\begin{align}
\left(\sum_{i\in[\alpha]}\frac{1}{i} \frac{(q-1)((q+1)^{\frac{i}{2}}-1)-q(q^{\frac{i}{2}}-1)}{q(q-1)}\right)-n &>0,\mbox{ i.e.} \\
\left(\sum_{i\in[\alpha]}\frac{1}{i} 
\frac{(q-1)(q+1)^{\frac{i}{2}}-qq^{\frac{i}{2}}+1}{q(q-1)}\right)-n &> 0.\label{eqfinal}
\end{align}
With $f(i)=\frac{(q-1)(q+1)^{\frac{i}{2}}-qq^{\frac{i}{2}}+1}{iq(q-1)}$ for all $i\in[\alpha]$, the proof of (\ref{eqfinal}) contains the following steps
\begin{enumerate}
\item\label{item1} For all $i\geq 2$ we have $f(i)\geq 0$,
\item\label{item2} $f(5)+f(1)\geq 0$,
\item\label{item3} $f(\alpha)-n > 0$.
\end{enumerate}
Step \ref{item1}.: For $i=2$ we have
\begin{align*}
f(2)
=\frac{1}{2}\frac{(q-1)(q+1)-q^2+1}{q(q-1)}
=\frac{q^2-1-q^2+1}{2q(q-1)}
=0.
\end{align*}
For $i=3$ we have 
\begin{align*}
f(3)
&=\frac{1}{3}\frac{(q-1)(q+1)\sqrt{q+1}-q^2\sqrt{q}+1}{q(q-1)}.
\end{align*}
Consider the function $g:\R\rightarrow\R;q\mapsto q^4-2q^3-2q^2+q+1$.
This function has two minima (between $-0.75$ and $-0.5$ as well as between $1.75$ and $2$) and one maximum (between $0.125$ and $0.25$). Since $g$ has only two inflexion points and $g$ is strictly greater than zero at the first minima, $g$ has only two roots. The first root is between $0.7$ and $0.8$ and the second root is between $2.5$ and $2.75$. Thus for all $q\geq 2.75$ we have $g(q)>0$. This implies $q^5+q^4-2q^3-2q^2+q+1>q^5$. Hence equivalently we get $(q+1)(q^4-2q^2+1)>q^5$, i.e., $(q+1)(q^2-1)^2>qq^4$.
This implies $\sqrt{q+1}(q^2-1)>\sqrt{q}q^2$ which proves that the numerator of $f(3)$ is positive and hence $f(3)>0$. Before we prove the claim for $i\geq 4$, we will prove that $(q-1)(q+1)^j\geq q^{j+1}$ for $j\geq 2$. 
Firstly we get
$$
(q-1)(q+1)^j
=\left(\sum_{k\in [j]}\left (\binom{j}{k-1}-\binom{j}{k}\right) q^k\right) +q^{j+1}-1.
$$
Due to the central symmetry of each row of the Pascal triangle and since the distribution of the binomial coefficient is unimodal, for $k\le \lfloor j/2\rfloor$, we have 
$$\binom{j}{j-k}-\binom{j}{j-k+1}=-\left(\binom{j}{k-1}-\binom{j}{k}\right)>0$$
and thus
$$
(q-1)(q+1)^j=\left(\sum_{k\in [\lfloor j/2\rfloor]} \left (\binom{j}{k}-\binom{j}{k-1}\right) (q^{j-k+1}-q^k)\right) +q^{j+1}-1.
$$
Since $k\le \lfloor j/2\rfloor$, we have $j-k+1>k$ and each term of the above sum is thus positive. 
This shows that $(q-1)(q+1)^j\ge q^{j+1}$.
This leads to the following estimations for $f(i)$.
For $i=2j$ and $j\geq 2$ we get
\begin{align*}
f(i)
&= \frac{(q-1)(q+1)^j-qq^j+1}{iq(q-1)}\geq \frac{q^{j+1}-q^{j+1}+1}{iq(q-1)}>0.
\end{align*}
Finally for $i=2j+1$ and $j\geq 2$ we get
\begin{align*}
f(i)
&=\frac{(q-1)(q+1)^j\sqrt{q+1}-qq^j\sqrt{q}+1}{iq(q-1)}
&\geq \frac{q^{j+1}(\sqrt{q+1}-\sqrt{q})+1}{iq(q-1)}>0.
\end{align*}
Step \ref{item2}.: Notice that $\alpha\geq 7$ holds and thus $f(5)$ is always a summand.
For $f(5)+f(1)$ we have to prove
\begin{align*}
\frac{(q-1)(q+1)^2\sqrt{q+1}-qq^2\sqrt{q}+1}{5q(q-1)} + \frac{(q-1)\sqrt{q+1}-q\sqrt{q}+1}{q(q-1)}\geq 0
\end{align*}
Thus we get for the numerator 
\begin{align*}
&(q-1)(q+1)^2\sqrt{q+1}-qq^2\sqrt{q}+1+5(q-1)\sqrt{q+1}-5q\sqrt{q}+5\\
&= (q-1)\sqrt{q+1}((q+1)^2+5)-q\sqrt{q}(q^2+5)+6\\
&=(q-1)\sqrt{q+1}(q^2+2q+6)-q\sqrt{q}(q^2+5)+6\\
&=q^3\sqrt{q+1}+q^2\sqrt{q+1}+4q\sqrt{q+1}-6\sqrt{q+1}-q^3\sqrt{q}-5q\sqrt{q}+6.
\end{align*}
We have $q^3 \sqrt{q+1} > q^3 \sqrt{q}$ and, since $q \geq 3$,
$$
q^2 \sqrt{q+1} \geq (6+q) \sqrt{q+1}.
$$
Therefore $q^2 \sqrt{q+1} + 4q \sqrt{q+1} \geq 6 \sqrt{q+1} + 5q \sqrt{q}$ and the numerator is positive.

\smallskip
Step \ref{item3}.: Notice that for fixed $i$, $f(i)$ is monotonically increasing for increasing $q$. This implies 
\[
f(\alpha)\geq \frac{2\cdot 4^{\frac{\alpha}{2}}-3\cdot 3^{\frac{\alpha}{2}}+1}{6\alpha}
=\frac{2^{\alpha+1}-3^{\frac{\alpha}{2}+1}+1}{6\alpha}.
\]
We are going to prove that 
\begin{align}\label{lastEquation}
2^{\alpha+1}-3^{\frac{\alpha}{2}+1}+1 > 6\alpha n.
\end{align}
 Recall that $\alpha$ is a function of $n$, given by $\alpha = \lfloor \frac{16}{7} \sqrt{n} \rfloor + 5$. 

First, we have
$$
2^{\alpha+1}-3^{\frac{\alpha}{2}+1} > 2^{\alpha-1} - 2^{\frac{\alpha}{2}}.
$$
Indeed, this inequality is equivalent to
\begin{align*}
2^{\frac{\alpha}{2}} \left(3 \cdot 2^{\frac{\alpha}{2}-1} + 1 \right) > 3^{\frac{\alpha}{2}+1} \quad \Leftrightarrow \quad
2^{\frac{\alpha}{2}-1} + \frac{1}{3} > \left(\frac{3}{2}\right)^{\frac{\alpha}{2}}.
\end{align*}
We can check that this last inequality is true by taking the logarithm of both sides, since $\alpha > 5$.

Therefore, it is sufficient for \eqref{lastEquation} to show that $$ 2^{\alpha-1} - 2^{\frac{\alpha}{2}} = 2^{\frac{\alpha}{2}} (2^{\frac{\alpha}{2}-1} -1) > 6 \alpha n.$$

Note that $2^{\frac{\alpha}{2}} > n$ (indeed, $ \lfloor \frac{16}{7} \sqrt{n} \rfloor+5 > 2 \sqrt{n}+5$, thus $2^{\frac{\alpha}{2}} > 2^{\frac{5}{2}} \cdot 2^{ \sqrt{n}}$). Once again, taking the logarithms, one can check that $2^{\frac{5}{2}} \cdot 2^{ \sqrt{n}} > n$ holds.

To verify \eqref{lastEquation}, it remains to show that $2^{\frac{\alpha}{2} - 1}  -1 \geq 6 \alpha$ or that
$
2^{\frac{\alpha}{2} - 1}  > 6 \alpha 
$.
Taking the logarithms, it is equivalent to
\begin{align*}
&\frac{\alpha}{2} - 1 > \log(6) + \log(\alpha) \\
\Leftrightarrow &\alpha - 2 \log(\alpha) > 2 \log(6) + 2,
\end{align*}
which is true for $\alpha \geq 15$, that is for $n \geq 16$. Equation~\eqref{lastEquation} can be verified by a computer for $q-1 \leq n < 16$.


\medskip

By \ref{item1}., \ref{item2}., and \ref{item3}. 
Equation (\ref{eqfinal}) is proven and this proves the claim.\qed
\end{proof}

\begin{remark}
Since the estimation in Theorem~\ref{reconstructgeneral} depends on the distribution of the letters in contrast to the method of reconstruction, it is wise to choose an order $<$ on $\Sigma$ such that $x<y$ if $|w|_x\leq |w|_y$. In the example we have chosen the {\em natural} order $\ta<\tb<\tn$ which leads in the worst case to fourteen binomial coefficients that has to be taken into consideration. If we chose the order $\tb<\tn<\ta$ the formula from Theorem~\ref{reconstructgeneral} provides that ten binomial coefficients suffice. This observation leads also to the fact that less binomial coefficients suffice for a unique determinism if the letters are not distributed equally but some letters occur very often and some only a few times.
\end{remark}

\begin{remark}
Let's note that the number of binomial coefficients we need is at most $qn$. Indeed, we will prove that $\sum_{i\in[q]}|w|_i(q+1-i)\leq qn$. We have $qn
=qn+n-n
=q\sum_{i\in[q]}|w|_i+\sum_{i\in[q]}|w|_i-\sum_{i\in[q]}|w|_i
\geq q\sum_{i\in[q]}|w|_i+\sum_{i\in[q]}|w|_i-\sum_{i\in[q]}(|w|_ii)
=\sum_{i\in[q]}|w|_i(q+1-i)$.
\end{remark}

\section{Conclusion}\label{conc}
In this paper we have proven that a relaxation of the so far investigated reconstruction problem from scattered factors from $k$-spectra to arbitrary sets yields that less scattered factors than the best known upper bound are sufficient to reconstruct a word uniquely. Not only in the binary but also in the general case the distribution of the letters plays an important role: in the binary case the amount of necessary binomial coefficients is smaller the larger $|w|_\ta-|w|_\tb$ is. The same observation results from the general case - if all letters are equally distributed in $w$ then we need more binomial coefficients than in the case where some letters rarely occur and others occur much more often. Nevertheless the restriction to right-bounded-block words (that are intrinsically Lyndon words) shows that a word can be reconstructed by fewer binomial coefficients if scattered factors from different spectra are taken. Further investigations may lead into two directions: firstly a better characterisation of the uniqueness of the $k_{\ta^{\ell}\tb}$ would be helpful to understand better in which cases less than the worst case amount of binomial coefficients suffices and secondly other sets than the right-bounded-block words could be investigated for the reconstruction problem.

\end{document}